\documentclass[runningheads]{llncs}

\usepackage[T1]{fontenc} 
\usepackage{graphicx}
\usepackage{xcolor}
\usepackage{mathtools}
\usepackage{amsmath}
\usepackage{commath}
\usepackage{enumerate}
\usepackage{mathrsfs}
\usepackage{tabularx}
\usepackage{tikz}
\usetikzlibrary{arrows.meta,positioning}
\usepackage{eurosym}
\usepackage{comment}
\usepackage{algpseudocode}
\usepackage{algorithm}
\usepackage{amssymb}
\usepackage
[disable]
{todonotes}
\usepackage{environ}
\usepackage{color}
\usepackage{array}
\usepackage{multirow}
\usepackage{xspace}

\newif\iflncs
\lncsfalse

\newcommand{\nop}[1]{}

\newcommand{\N}{\mathbb{N}}

\newcommand{\DelayP}{\mathsf{DelayP}}
\newcommand{\DelayFPT}{\mathsf{DelayFPT}}
\newcommand{\fpt}{\mathsf{FPT}}
\newcommand{\FPT}{\mathsf{FPT}}
\newcommand{\fptt}{\mathsf{FPT}\text{-time}}
\newcommand{\fptd}{\mathsf{FPT}\text{-delay}}

\newcommand{\ptime}{\mathsf{P}}

\newcommand{\fvst}{\text{\sc Feedback Vertex Set in Tournaments}\xspace}

\newcommand{\gfvst}{\text{\sc GFVST}\xspace}
\newcommand{\vc}{\text{\sc Vertex Cover}\xspace}
\newcommand{\clst}{\text{\sc Closest String}\xspace}
\newcommand{\clstp}{\text{\sc Closest String with Prefix}\xspace}

\newcommand{\lpth}{\text{\sc Longest Path}\xspace}

\newcommand{\ilplong}{\text{\sc Integer Linear Programming with k Variables}\xspace}

\newtheorem{defi}[definition]{Definition}%

\newtheorem{prop}[proposition]{Proposition}

\begin{document}
\title{From FPT Decision to FPT Enumeration} 
\author{Nadia Creignou
\inst{1} 
\and
Timo Camillo Merkl
\inst{2}
\and
Reinhard Pichler
\inst{2} %
\and
Daniel Unterberger
\inst{2} %
}

\authorrunning{N. Creignou et al.}

\institute{Aix Marseille Univ, CNRS, LIS, Marseille, France \\
\email{nadia.creignou@univ-amu.fr}
\and TU Wien, Vienna, Austria \\
\email{\{timo.merkl|reinhard.pichler|daniel.unterberger\}@tuwien.ac.at}}
\maketitle              %
\begin{abstract}
Fixed-parameter tractable (FPT) algorithms have been successfully applied 
to many intractable problems -- with a focus on decision and optimization problems.
Their aim is to confine the exponential explosion to some parameter,
while the time complexity only depends polynomially on the instance size.
In contrast, intractable enumeration problems 
have received comparatively little attention so far. 
The goal of this work is to study how FPT decision algorithms could be turned into FPT enumeration algorithms. We thus inspect several fundamental approaches for designing FPT decision or optimization algorithms and we present ideas how they can be 
extended to FPT enumeration algorithms. 
\end{abstract}
\section{Introduction}

Many practically relevant computational problems
in a wide range of applications 
are intractable.
A meanwhile well-established 
approach of dealing with intractable problems is 
to aim at fixed-parameter tractable ($\fpt$) algorithms 
(see \cite{DBLP:books/sp/CyganFKLMPPS15,DBLP:books/ox/Niedermeier06} for comprehensive collections of 
$\fpt$-results). That is, the exponential explosion can be 
confined to some parameter (typically some structural property of the problem instances 
or of the desired solution) while the time complexity only depends polynomially on 
the size of the problem instance.

$\fpt$-algorithms have been mainly developed for {\em decision problems} (e.g., does there exist a solution?) or {\em optimization problems} (e.g., what is the max/min value attainable by all solutions?). However, in many settings, the user might be interested in {\em all} (or at least a significant portion) of the solutions or optimal solutions to a problem. This is particularly the case when querying data or knowledge base systems, where the user is usually interested in retrieving all (or at least a significant portion) of the answers
rather than just a simple ``yes''. We are thus faced with an {\em enumeration problem}, whose task is to output the entire solution set without duplicates 
(i.e., outputting the same solution twice).
Note that even
seemingly simple problems may have a huge (potentially  exponential) number of solutions.
Hence, as was shown in the foundational paper by 
Johnson et al.~\cite{DBLP:journals/ipl/JohnsonP88}, even the 
notion of tractability has to be redefined. The enumeration complexity class that 
is often considered as an analog of the decision complexity class $\ptime$ is $\DelayP$, 
i.e., ``polynomial delay''. It 
consists of those problems which have an algorithm that requires polynomial time 
in the size of the instance to output the first solution, to output any further solution, 
and to terminate after the last solution.
The study of enumeration problems and the development of $\DelayP$-algorithms has received a lot 
of attention, 
especially in the areas of graph theory and databases, see 
e.g., \cite{DBLP:conf/sigmod/BerlowitzCK15,%
ConteGMR2018,%
DBLP:conf/stoc/ConteU19,%
DBLP:journals/siamdm/KanteLMN14,%
DBLP:conf/pods/LivshitsK17}.

But what happens if the task of finding the first or a further solution is intractable? 
A natural approach
is to borrow ideas from decision and optimization problems and to aim for
$\fptd$, i.e., algorithms that find the first or next solution 
in $\fpt$-time. Surprisingly, this line of research has received little attention so far;  
notable exceptions are 
\cite{DBLP:journals/algorithms/CreignouKMMOV19,%
DBLP:journals/corr/CreignouMMSV13,%
damaschke06,fernau02,%
DBLP:journals/jcss/GolovachKKL22,%
home-made:Meier20}. 
The aim of our work is to shed some light on this under-researched area and to 
provide assistance to algorithm developers in designing $\fptd$ enumeration algorithms.
To this end, we revisit the most common techniques in 
the design of $\fpt$ decision and optimization algorithms 
and look for their extension to 
$\fptd$ enumeration algorithms. 
For the basic $\fpt$ approaches considered here, we will mainly follow the 
presentation in~\cite{DBLP:books/sp/CyganFKLMPPS15}. Apart from kernelization
\cite{DBLP:journals/corr/CreignouMMSV13,DBLP:journals/jcss/GolovachKKL22,DBLP:conf/iwpec/BougeretGSS25}
and tree decompositions (in particular, enumeration algorithms via
Courcelle's Theorem~\cite{DBLP:journals/tcs/CourcelleM93}), these techniques have not yet made their way into enumeration.
Another $\fpt$ approach, that has already been extended to enumeration~\cite{DBLP:journals/corr/CreignouMMSV13} 
but is not mentioned in~\cite{DBLP:books/sp/CyganFKLMPPS15}, are backdoors. In particular, for various variants of the SAT-problem,
they have received quite some attention~\cite{DBLP:journals/jcss/DreierOS24,%
DBLP:conf/focs/GaspersS13,DBLP:conf/sat/MisraORS13}.

Among the enumeration methods that we apply are adaptations and extensions 
of typical enumeration approaches 
\cite{DBLP:phd/it/Marino12}
such as ``flashlight algorithms'' (which search for solutions in a top-down traversal of 
some tree structure and repeatedly use a 
decision algorithm  as a ``flashlight''
to see if a certain branch
contains solutions before expanding it) or exploiting the closure
under union of $\DelayP$ \cite{home-made:Strozecki10}.
A particular
challenge of enumeration problems comes from the requirement that duplicates 
have to be avoided. Here, we strongly draw inspiration from the landmark paper 
of Lawler~\cite{Lawler72}, where a basic technique for avoiding duplicates by 
appropriate partitioning of the solution space was presented.

The paper is organized as follows: in Section \ref{sec:preliminaries}, we recall some 
fundamental definitions of parameterized enumeration. The main part of the paper is 
devoted to the design of $\DelayFPT$ algorithms by carrying over basic $\fpt$-methods
from decision algorithms, namely partition algorithms as an 
extension of bounded search trees (in Section \ref{sec:partition}),
enumeration algorithms for a union of (not disjoint!) ``$\fpt$-many'' subsets of solutions
(in Section~\ref{sec:UnionEnumeration}), iterative compression-based enumeration 
algorithms (in Section~\ref{sec:IterativeCompression}) and dynamic programming algorithms 
(in Section~\ref{sec:DynamicProgramming}). The design of enumeration 
algorithms via these $\fpt$-methods is illustrated by instructive examples, and
we propose formalizations to generalize these ideas.
We conclude with Section~\ref{sec:Conclusion}.
Further details 
are  provided in
\iflncs
\cite{DBLP:journals/corr/abs-2509-11929}.
\else
the appendix.
\fi

\section{Parameterized Complexity and Enumeration}\label{sec:preliminaries}

We assume familiarity with basic notions of complexity and graph theory. 
A \emph{parameterized problem} is defined by a ternary predicate $Q \subset \Sigma^* \times \N \times \Sigma^*$ for some alphabet $\Sigma$. A tuple $(I,k) \in \Sigma^* \times \N$ is  an instance, 
$k$ is its parameter, and $Q(I,k) \coloneqq \{S \in \Sigma^* \mid (I,k,S) \in Q\}$ 
is its solution set.
The task of deciding whether $Q(I,k) \neq \emptyset$ 
is the {\em decision problem $Q$}, while the task of enumerating $Q(I,k)$ without duplicates
is the {\em enumeration problem $Q$}. 
An instance $(I,k) \in \Sigma^*\times \N$ with $Q(I,k) \neq \emptyset$ is a \textit{yes-instance}, otherwise a \textit{no-instance}. We will denote the size of an instance $(I,k)$ by $|(I,k)|$.

\begin{defi}
    Let $Q$ be a decision problem. We say $Q$ can be decided in $\fptt$ if there exists an algorithm $\mathcal{A}$ that, for any instance
    $(I,k) \in \Sigma^* \times \N$, decides if $Q(I,k)\neq \emptyset$ with a runtime bounded by $h(k)p(|(I,k)|)$ for some computable function $h \colon \N \to \N$ and polynomial~$p$. We call the class of all problems admitting $\fptt$ algorithms $\fpt$.
\end{defi}

An \emph{enumeration algorithm} $\mathcal{A}$ for an enumeration problem $Q$ is an
algorithm which, when given an instance $(I,k)$  of $Q$ as input, 
outputs the elements of $Q(I,k)$
without duplicates.
An enumeration problem is in the class $\DelayP$, if it has an algorithm with a polynomial bound 
on its \emph{delay}. 
By ``delay'' we mean the maximum time between  two consecutive solutions, including  
the time until the first solution is output and  the time after the last solution is output.
The natural parameterized analog is $\DelayFPT$, which we define next.

\begin{defi}
    Let $Q$ be an enumeration problem. We say $Q$ can be enumerated with \emph{$\fptd$} if there exists an algorithm $\mathcal{A}$ which, for any instance $(I,k) \in \Sigma^* \times \N$, enumerates all solutions in $Q(I,k)$ with a delay bounded by $h(k)p(|(I,k)|)$ for some computable function $h:\N \to \N$ and polynomial $p$. We call the class of all problems $Q$ admitting such an $\fptd$ enumeration  algorithm $\DelayFPT$.
\end{defi}

\section{Partition Algorithms}\label{sec:partition}

Search tree algorithms, which work by iteratively splitting a given problem instance into sub-instances, 
are among the most commonly used techniques in $\fpt$ algorithms 
for decision problems. The correctness of such algorithms hinges on the fact that, in case of a yes-instance,
at least one of the sub-instances resulting from a splitting step leads to a feasible solution. 
In case of enumeration, we have to be careful that 
the splitting step {\em partitions} the set of solutions, so as to preserve all solutions 
and avoid duplicates. 
We therefore start our discussion of $\DelayFPT$ enumeration algorithms by studying various forms of \emph{partition algorithms}.

\subsection{Bounded Search Tree Partition Algorithms}
\label{subsec:boundedSearchTree}

We know from decision algorithms
that, 
if the total size of the search tree is bounded by a function of the parameter only and every step takes $\fptt$, 
then such an algorithm runs in $\fptt$. 
To carry these ideas over to enumeration algorithms, several extensions are required. 
On one hand, we still need bounds (that only depend on the parameter $k$) on the breadth and depth of the search tree.
In addition, following the philosophy of \cite{Lawler72}, we now have to make sure that, with every splitting step, (i) {\em every} solution is 
contained in some branch, (ii) {\em no} solution is contained in {\em more than one branch}, and (iii) the solutions of the 
{\em instances at the leaf nodes} of the search tree can be enumerated with $\DelayFPT$.

We illustrate these ideas on the 
\text{\sc Feedback Vertex Set in Tournaments} 
(FVST, for short)
problem.
For a directed  graph $G=(V,A)$, we call $S \subseteq V$ a \emph{feedback vertex set}, if $G[V\setminus S]$ (that is the  subgraph of $G$ induced by the vertex set $V\setminus S$) is acyclic. A \emph{Tournament} $T=(V,A)$ is a directed graph 
such that, for any $u$ and $v$ in $V$, exactly one of $(u,v)$ and $(v,u)$ is in $A$.
In the FVST enumeration problem, we are given a tournament $T = (V,A)$
and a parameter $k\in \N$, and we want to enumerate all feedback vertex sets of size at most $k$. 

The FVST decision problem can be shown to be in $\fpt$ via bounded search trees
(see \cite[Section 4.2]{DBLP:books/sp/CyganFKLMPPS15}). Below we show how we can get 
an enumeration 
algorithms satisfying the above requirements (i) -- (iii) 
by maintaining additional information.

\begin{proposition}\label{prop:fvst}
The $\fvst$ problem is in $\DelayFPT$ using polynomial space.
\end{proposition}

\begin{proof}
Our algorithm for enumerating the feedback vertex sets of a tournament $T$ maintains the problem instance $(T,k)$ together 
with two vertex sets $C, F\subseteq V(T)$ as additional information. The intended meaning of these vertex sets is as follows:
We store in $C$ vertices of $T$ that we definitely want to delete and in $F$ the vertices that we definitely want to retain. 
Hence, our enumeration algorithm recursively searches for 
feedback vertex sets $U$ in  the tournament $T[V\setminus C]$ 
having $U \cap F = \emptyset$. Initially, we set $F = C = \emptyset$.

A tournament has a directed cycle if and only if it has a directed triangle (see \cite[Lemma 4.3]{DBLP:books/sp/CyganFKLMPPS15}). 
We can find such a triangle  (or confirm that none exists) in time $O(|V|^3)$. If none exists, 
then we can output all feedback vertex sets $S$ of size $\leq k$ with $C \subseteq S$ and 
$S \cap F = \emptyset$
with polynomial delay.
Otherwise, since each feedback vertex set must have non-empty intersection with this triangle, we branch 
by adding any non-empty subset of the triangle vertices disjoint from $F$ to $C$ 
($\leq 7$ possibilities), reduce $k$ accordingly, and add 
the remaining triangle vertices to $F$. We recursively apply these steps, i.e., finding a directed triangle and branching over non-empty intersections. As we only need to store the path of the currently examined recursion tree 
(in addition to the children with possibly some arbitrary ordering of them for all nodes on the path) 
and the depth is at most $k$, we  only need polynomial space.

We verify that  the requirements on a bounded search  tree partition algorithm are fulfilled: the breadth of the search tree
is bounded by 7 and the depth is bounded by $k$ (since the parameter is strictly decreased in each sub-instance). 
Clearly, the splitting preserves all solutions and the additional information in the sets $F,C$ excludes duplicates. 
Finally, for a sub-instance with $k = 0$, we know that the solution set is empty if the subgraph of $T$ induced by $V \setminus C$ 
contains a cycle; otherwise we output the solution $S\supseteq C$ as described above.
\end{proof}

\noindent
{\bf Generalization and formalization.}
Intuitively, we need a method $f$ to create child instances of a given instance such that the solution space is partitioned by the solution space of the child instances. We recursively apply $f$ until the instances are fully processed, i.e., until all solutions can be easily enumerated with a method $f_0$. To limit the size of the possible search tree and to  formalize when instances are fully processed, we need a function $\gamma$, which assigns each instance $(I,k)$ a natural number $m$ such that the tree rooted at $(I,k)$ has a height of at most $m$ (where $m$ can, in turn,
only depend on the parameter $k$),
and $m=0$ if and only if the instance is fully processed. The following definition is intended to  capture these requirements.

\begin{defi}\label{def:bounded_tree_partition}
    For an enumeration problem $Q$, we call a tuple of functions $(f,\gamma)$ a \emph{bounded search tree partition algorithm} if, for some computable functions $b,h\colon \mathbb{N} \rightarrow \mathbb{N}$ (for breadth and height, respectively), the following conditions hold: 
    \begin{itemize}
        \item $f \colon \Sigma^* \times \N \to \mathcal{P}(\Sigma^* \times \N)$ is an $\fptt$  function such that for all $(I,k)\in \Sigma^* \times \N $: 
        \begin{itemize}
            \item $|f(I,k)| \leq b(k)$,            
            \item for all $(I',k') \in f(I,k)$, $|(I',k')| \leq |(I,k)|$ and $k' \leq k$,
            \item $Q(I,k) =\biguplus_{(I',k') \in f(I,k)} Q(I',k')$, where $\biguplus$ denotes the disjoint union, 
        \end{itemize}
        \item $\gamma\colon \Sigma^* \times \N \to \N$ is an $\fptt$ function such that for all $(I,k)\in \Sigma^* \times \N$: 
        \begin{itemize}
            \item $\gamma(I,k) \leq h(k)$, %
            \item if $\gamma(I,k) \neq 0$, then for all $(I',k') \in f(I,k)$, $\gamma(I',k') < \gamma(I,k)$.
        \end{itemize}
    \end{itemize}
     Additionally, we require that there must exist a  function $f_0$, which enumerates $Q(I,k)$ with $\fptd$ when $\gamma(I,k) = 0$.
\end{defi}
As we will see in the rest of this section, the condition that both breadth and height are bounded 
by a function of the parameter only, can be quite restrictive.
 
In the definition above, the first condition on $f$ controls the breadth of the search tree,
while the function $\gamma$ is a measure that decreases at each step, 
and thus bounds the depth of the search tree. 
Both are bounded by a function of $k$, which thus also applies to 
the size of the search tree.
The two last conditions on $f$ ensure that traversing the search tree will provide all solutions with no repetitions. Therefore, we get the following result.

\begin{prop}
\label{prop:bounded_search_tree}
    Let $Q$ be a parameterized enumeration problem. If $Q$ has a 
    bounded search tree  partition algorithm $(f,\gamma)$, then $Q$ is in $\DelayFPT$.
\end{prop}

To fit our FVST algorithm from the proof of Proposition \ref{prop:fvst}
 into this general framework, we consider a generalized version of FVST (we will call it GFVST) where
instances are extended by sets $C$ and $F$. Then the FVST problem corresponds to the 
special case of GFVST with $C = F = \emptyset$. Function $f$ defines one splitting step, where an instance of GFVST gives rise to up to $b(k) = 7$ new instances.
The function $\gamma$ is initially set to $k$ and decreased by the number of vertices
added to $C$ in each splitting step. To trigger the application of function $f_0$ for outputting all solutions of a given instance, we set $\gamma$ to 0 if no more directed triangle exists in the graph. Of course, $\gamma$ also becomes 0 when the
``budget'' of allowed vertex deletions has been exhausted.
For further details,  
\iflncs
see~\cite{DBLP:journals/corr/abs-2509-11929}.
\else
see~Appendix~\ref{app:boundedSearchTree}.
\fi

\subsection{Flashlight Partition Algorithms}
\label{subsec:Flashlight}

Next, we aim at relaxing 
the bound on the size
of the search tree.
More specifically, we bound both the breadth and the depth by an $\fpt$-function rather than by a function of the parameter $k$ alone. 
To guarantee that we find a new solution after $\fptt$, we are only allowed to explore a branch in the search tree if it leads to a feasible solution and we have to go through the tree in a depth-first manner. 
This can be guaranteed by using a ``flashlight'' algorithm, i.e., a decision algorithm that decides in $\fptt$ whether an instance has a solution or not. 
Again, we need to ensure that the branching properly partitions the solution space and that solutions of instances at leaf nodes can be enumerated with $\DelayFPT$.

We illustrate the idea of flashlight partition algorithms by considering the enumeration variant of 
the $\clst$ problem, which is defined as follows:  
we are given a set of strings $X = \{x_1, \dots, x_n\}$ of length $L$ over an alphabet $\Sigma$ and a parameter $k \in \N$. The task is to find all strings $s$ of length $L$ such that $d_H(s,x_j) \leq k$ for 
each
$1 \leq j \leq n$, where $d_H(x,y)$ is the Hamming distance between strings $x$ and $y$.
We call such a string $s$ a \emph{center string}.

$\fpt$-membership of the decision variant of $\clst$ (i.e., does there exist a center string?)
can be shown by a bounded search tree algorithm (see \cite[Section 3.5]{DBLP:books/sp/CyganFKLMPPS15}). 
However, for an enumeration algorithm, we again need to maintain more information to make sure that
every splitting of a problem instance {\em partitions} the solution set.

\begin{prop}\label{prop:closet_string}
The $\clst$ problem is in $\DelayFPT$ using $\fpt$ space.    
\end{prop}

\begin{proof}
Our algorithm for enumerating the center strings of a given instance $(X,k)$ 
maintains, in addition to the instance, a string $\omega \in \Sigma^*$.  The idea of $\omega$ is that
we only search for center strings with prefix $\omega$. Initially, we set $\omega$ to the empty string.
As flashlight, we use an $\fpt$-algorithm for the   
\clst decision problem, checking if there exists a center string with prefix $\omega$.
Such an $\fpt$-algorithm is
easily obtained by adapting the bounded search tree algorithm from 
Section 3.5 in \cite{DBLP:books/sp/CyganFKLMPPS15}.

Now let $X$ be an instance of the 
\clst enumeration problem together with a prefix $\omega$ as additional information.
We split this problem instance into sub-instances by  iterating through all possible characters $\sigma\in \Sigma$
and tentatively generating all possible prefixes $\omega\cdot\sigma$.
However, 
from these, we only keep the yes-instances (which can be determined in $\fptt$ 
by the flashlight). 

The depth of the search tree at an instance $(X,k)$ with additional information $\omega$ is therefore $|\omega|\leq L$,
where $L$ the length of the strings in $X$. The breadth is bounded by the size of the alphabet $\Sigma$. Clearly, 
fixing the prefix of solutions partitions the solution space. 
Moreover, no solution gets lost if we only 
keep those sub-instances that  contain a solution according to the flashlight.
Thus, in total, by traversing this search tree in a depth-first fashion and outputting the solutions at the leaves, we get an $\fptd$ using $\fpt$ space.
\end{proof}

\noindent
{\bf Generalization and formalization.}
As with bounded search tree partition, 
we need to make sure that each step can be done in $\fptt$
and 
partitions the solution space, and that the solutions at the leaf nodes can be enumerated with $\fptd$. 
The maximal depth $\gamma$ may now also depend  on $|I|$, with the dependence being polynomial. 
An $\FPT$-decision ``flashlight'' algorithm is thus needed to check the existence of solutions.
We formalize these requirements in the following definition.

\begin{defi}\label{def:flash_partition}
    For an enumeration problem $Q$, we call a tuple of functions $(f,\gamma)$ a \emph{flashlight partition algorithm} 
    if, for some computable function $h\colon \N \to \N$ and polynomial $p$, the following conditions hold:
    \begin{itemize}
        \item $f \colon \Sigma^* \times \N \to  \mathcal{P}(\Sigma^* \times \N)$ is an  $\fptt$ function such that for all $(I,k)\in \Sigma^* \times \N$: 
        \begin{itemize}
            \item for all $(I',k') \in f(I,k)$,  $|(I',k')| \leq |(I,k)|$ and $k' \leq k$,
            \item for all $(I',k') \in f(I,k)$, $Q(I',k') \neq \emptyset$, 
            \item $Q(I,k) = \biguplus_{(I',k') \in f(I,k)} Q(I',k')$,
        \end{itemize}
        \item $\gamma\colon \Sigma^* \times \N \to \N$ is an $\fptt$ function such that for all $(I,k)\in \Sigma^* \times \N$: 
        \begin{itemize}
            \item $\gamma(I,k) \leq h(k)p(|(I,k)|)$, %
            \item if $\gamma(I,k) \neq 0$, then for all $(I',k') \in f(I,k)$, $\gamma(I',k') < \gamma(I,k)$.
        \end{itemize}
    \end{itemize}
     Additionally, we require that there must exist a  function $f_0$, which enumerates $Q(I,k)$ with $\fptd$ when $\gamma(I,k) = 0$.
\end{defi}

The ``flashlight'' is implicit in the function $f$, which (in $\fptt$) 
only generates subproblems $(I',k')$ with non-empty solution set (by the second condition on $f$
in Definition~\ref{def:flash_partition}).
The depth of the search tree is bounded by an $\fpt$-function because of the measure $\gamma$ that decreases at each step along a branch. Together with the ``flashlight'', which guarantees that 
no branch of the search tree is explored in vain, we thus get $\fptd$.
We note that this method is related to the algorithmic paradigm of self-reducibility studied 
 in \cite{DBLP:journals/corr/CreignouMMSV13}.

\begin{prop}\label{prop:flashlight_partition}
    Let $Q$ be a parameterized enumeration problem. If $Q$ has a flashlight partition algorithm $(f,\gamma)$, then $Q$ is in $\DelayFPT$.
\end{prop}

Again, to fit our $\clst$  algorithm 
from the proof of Proposition~\ref{prop:closet_string}
into this framework, we consider a generalized version of this problem
where instances are extended by a prefix $\omega$. The $\clst$ problem
then corresponds to the special case with an empty prefix $\omega$. Crucially, 
the function $f$ now serves two purposes: on one hand, it splits the current instance into new instances by appending one more character to $\omega$; on the other hand, it serves as ``flashlight'' and only retains those new instances which have a solution that extends the current prefix. 
Initially, $\gamma$ is set to $L$. It is then decremented in every splitting step,
when we extend the prefix by one character. 
The function $f_0$ for enumerating all solutions of an instance 
when $\gamma$ reaches the value 0 is particularly simple in this case: we just need 
to output the prefix $\omega$ (which now has length $L$) of the current instance.
For further details, see
\iflncs
\cite{DBLP:journals/corr/abs-2509-11929}.
\else
Appendix~\ref{app:Flashlight}.
\fi

\subsection{Solution Search Partition Algorithms}
\label{subsec:SolutionSearch}

We now want to 
remove the restriction on the depth of the search tree
altogether.
The idea is to alternately perform two steps, namely (1) find a solution and (2) split the 
remaining search space. Intuitively, we thus try to find a spanning tree of the solution space.
As before, 
 the algorithm consists in traversing a tree. However, 
unlike the two previous methods, this tree is not implicitly  made up of sub-instances and partial solutions, with final solutions only at the leaves or computable from the leaves. 
Now, each node of the tree constitutes an instance from which a solution $S$ of the original instance can be produced
in $\fptt$.
Clearly, in the splitting step, we need to create sub-instances that partition the solutions of $(I,k)$ {\em excluding the solution $S$}.
We repeat this branching, exploring the search tree depth-first, until we arrive at a no-instance, at which point we consider the remaining branches.

To maintain $\fptd$ between the output of two successive solutions, we can use the ``alternative output'' technique (see \cite{Uno2003,DBLP:phd/it/Marino12}), which prevents the backtracking through long paths without outputting any solution: it  
outputs the solutions {\em before} the recursive call when the current depth-first search level is even, 
and {\em after} the recursive call when it is odd.

We now construct a solution search partition algorithm for  
{\sc Integer Linear Programming with k Variables} parameterized by the number of variables:
let $k$ and $m$ be positive integers, let $A$ be an $m\times k$ matrix with integer coefficients and $b\in \mathbb{Z}^m$, and let $k$ be the parameter. 
The task is to enumerate all vectors $x\in \mathbb{Z}^k$ satisfying the system of inequalities $Ax\le b$.
Note that the number of solutions is potentially infinite.
To ensure that we can output each solution in 
$\fptt$ (and implicitly eliminate the problem of an infinite solution space), 
we assume that the absolute values of the components of a solution vector 
are bounded by some function $2^{h(k)p(L)}$, where $L$ is the instance-size, 
$h$ is a computable function, and $p$ is a polynomial.
Otherwise, it could happen that simply outputting a single solution cannot be done in $\fptt$ due to its size.
This can be easily enforced by adding appropriate entries to $A$ and $b$. 
Crucially, one can find in $\fpt$-time (w.r.t.\ $k$) 
one solution vector where all components 
have an $\fpt$-bounded binary representation (see \cite[Theorem 6.5]{DBLP:books/sp/CyganFKLMPPS15}).
We denote this algorithm by $\mathcal{A}$.

\begin{prop}\label{prop:ilp}
$\ilplong$ 
is in \\$\DelayFPT$.   
\end{prop}

\begin{proof}
We utilize the $\fptt$ algorithm $\mathcal{A}$, which either finds a solution vector where all components have $\fpt$-bounded size (in binary representation) or detects that no such solution exists.
If a solution has been found, say $(s_1, \dots, s_k)$, then we create new instances by adding (in)equalities of the form $(x_1 = s_1)\land \cdots \land (x_i = s_i)\land (x_{i+1} > s_{i+1})$ or $ (x_1 = s_1)\land \cdots \land (x_i = s_i)\land (x_{i+1} < s_{i+1})$ for some $0\leq i < k$. 
This produces $2k$ sub-instances  for any given instance.
Moreover, these sub-instances clearly partition the solution space of a given problem instance while excluding 
the current solution $(s_1, \dots, s_k)$.
For each variable, we only need to keep up to two (in)equalities of this form (i.e., the most restrictive ones). 
Actually, we can ensure that the size of the considered instance never increases 
by initially adding the two inequalities $-2^{h(k)p(L)} < x_i < 2^{h(k)p(L)}$ for each $i$. 
Hence, the number of (in)equalities in any given instance explored through our algorithm is at most $2k$ more than in the original instance.
In other words, the size of all sub-instances thus produced is polynomially bounded in the size of the input problem instance. 
Hence, by applying the ``alternative output'' technique (see \cite{Uno2003,DBLP:phd/it/Marino12}) recalled above, we thus achieve
$\fpt$ delay,  
even without a bound on the depth of the search tree.
\end{proof}

\smallskip
\noindent
{\bf Generalization and formalization.}
We now remove the depth bound $\gamma$ entirely, and thus the focus is on finding a method by which, for any given solution $S$, we can find (and calculate in $\fptt$) child instances for which the solution spaces are disjoint and exclude precisely $S$. We need a method $f$ 
to generate these child instances, but also a method $g$ to find a solution to any given instance. Also, if no solution exists, we want $f$ to return the empty set.
The following definition captures these requirements.

\begin{defi}\label{def:solution_partition}
    For an enumeration problem $Q$, we call a tuple  of  functions $(f,g)$ a \emph{solution search partition algorithm}, if the following conditions hold: 
    \begin{itemize}
        \item $g\colon \Sigma^* \times \N \to \Sigma^*$ is an $\fptt$ function such that for all  $(I,k)\in  \Sigma^* \times \N$:
       \begin{itemize}
            \item  $g(I,k)$ computes a solution $S \in Q(I,k)$, if it exists.
       \end{itemize}
        \item $f\colon \Sigma^* \times \N  \times \Sigma^* \to  \mathcal{P}(\Sigma^* \times \N)$ is an $\fptt$ function such that for all  $(I,k,S)\in  \Sigma^* \times \N  \times \Sigma^*$:
 \begin{itemize}
        \item if $Q(I,k)=\emptyset$, then  $f(I,k,S)=\emptyset$,
        \item $Q(I,k)\setminus \{S\} = \biguplus_{(I',k') \in f(I,k,S)} Q(I',k')$,
        \item for all $(I',k') \in f(I,k,S)$,  $|(I',k')| \leq |(I,k)|$ and $k' \leq k$,
        \end{itemize}
    \end{itemize}
\end{defi}
Note that, due to the need to store the path we are exploring, which is unbounded in size, there is also no bound on the space requirement of such an algorithm.

\begin{prop}\label{prop:solution_partition}
    Let $Q$ be a parameterized enumeration problem. If $Q$ has  a solution search partition algorithm $(f,g)$, then $Q$ is in $\DelayFPT$.
\end{prop}

For \ilplong, 
$g$ is the original $\FPT$ algorithm, which we can apply directly as we do not need to generalize the problem setting for our enumeration. Function $f$ is used to create the $2k$ sub-instances with added inequalities. It is also responsible for 
eliminating redundant (in)equalities to make sure that 
we never get an instance with more than two additional (in)equalities per variable
(compared with the original system of inequalities).
The condition $k' \leq k$ on instances produced by 
function $f$ is trivially satisfied since we never add a variable.

\section{Union Enumeration}
\label{sec:UnionEnumeration}

In Section \ref{sec:partition}, we discussed $\fptd$ algorithms that partition the solution space into parts that can  be efficiently enumerated individually. 
As we have outlined, for some problems, this can be naturally achieved.
However, for other problems and other classical techniques to design $\fpt$ decision algorithms, preventing this overlap is less natural (e.g., in the problem considered below).
Thus, we next discuss a method for dealing  with non-empty intersections 
of some parts.

More concretely, suppose that the solution space $\mathcal{S} = \bigcup_i \mathcal{S}_i$
is given as a union, such that  the $\mathcal{S}_i$'s are not 
pairwise disjoint. 
It would be tempting to enumerate one subset $\mathcal{S}_i$ after the other and, before outputting 
the next solution $S \in \mathcal{S}_i$, 
check if it was already output when processing $\mathcal{S}_1, \dots, \mathcal{S}_{i-1}$. However, 
it may happen that we  encounter a long sequence of already output solutions and, therefore,  
$\fptd$ is no longer guaranteed. Instead, we use a technique developed in \cite[Section 2.5]{home-made:Strozecki10} for enumerating elements of a union of sets without duplicates. 
Besides the efficient enumerability of each set in the union (in our case with $\fptd$), this technique requires membership checks to be efficiently doable (in our case in $\fpt$).
Thus, intuitively, we seek to first split the solution space into $\fpt$-many subsets of solutions such that, (i) \textit{every} solution is part of some subset, (ii) every subset can be \textit{enumerated} with $\fptd$, and (iii) \textit{checking} whether a solution is part of a subset can be done in $\fptt$.

A clear application of this method is for derandomized color coding algorithms. 
Thus, we will illustrate the method by considering the $\lpth$ problem, which asks for \textit{simple paths}, i.e., paths that contain every vertex at most once, of length $k$ in a given graph $G$, where $k$ is the parameter. We here consider the length of a simple path to be the number of vertices it contains instead of the number of edges.

\begin{prop}
\label{prop:lpth}
    \lpth is in $\DelayFPT$.
\end{prop}

\begin{proof}
We are given a graph $G = (V,E)$ and an integer $k$.
The idea of color coding introduced in \cite{DBLP:journals/jacm/AlonYZ95} applied to $\lpth$ is the following:
\begin{enumerate}
    \item Introduce $k$ colors and color each vertex $v\in V$ in one color $\gamma(v)\in \{1,\dots,k\}$.
    \item For $C= \{1,\dots,k\}$, find a $C$-\textit{colorful} path in $G$ for $\gamma$, i.e., a path where every color $i\in C$ appears exactly once on the path. 
\end{enumerate}
We first discuss how to find all colorful paths in $G$ under a given coloring $\gamma\colon V \rightarrow \{1,\dots, k\}$.
For that, we follow a dynamic programming approach and compute, for every $C\subseteq \{1,\dots, k\}$ and $u\in V$,
the following set $P(C,u)$:
\begin{align*}
    P(C,u) := \{v\in V {} \mid {} & \text{there exists a } C\text{-colorful path } (v_1,\dots,v_{|C|}) \text{ in } G \text{ such that } \\
    & v_{|C|-1} = v, v_{|C|} = u, \text{ and } C = \{\gamma(v_1), \dots, \gamma(v_{|C|})\}\}.
\end{align*} 
That is, for a vertex $u$ and subset of colors $C\subseteq \{1,\dots, k\}$, we keep track of those vertices $v$ that are the predecessors of $u$ on $C$-colorful paths that end in $u$.
Then, clearly, $P(\{1,\dots, k\},u) \neq \emptyset$ if and only if there is a colorful path (of length $k$) that ends in $u$.

To compute all sets $P(C,u)$, we start by iterating through the edges $vw\in E, \gamma(w)\neq \gamma(v)$ and add $v$ to $P(\{\gamma(w), \gamma(v)\}, w)$ (symmetrically for $w$ as well).
Then, for $|C|>2$ and $ \gamma(u)\in C$, we compute $P(C,u)$, using the equality 
\begin{align*}
    P(C,u) = \{v\in V \mid vu\in E \text{ and } P(C\setminus \{\gamma(u)\}, v)\neq \emptyset\}.
\end{align*} 
Thus, in total, we can compute all $P(C,u)$ in $\fpt$-time.
Moreover, we can enumerate all $\{1,\dots,k\}$-colorful paths in $G$ for $\gamma$ with $\fptd$ by using these sets.
Concretely, we can simply use $k$ nested loops: The first loop iterates through all $v_1\in V$ such that $P(\{1,\dots,k\},v_1)\neq \emptyset$. 
The second loop iterates through all $v_2\in P(\{1,\dots,k\},v_1)$.
The $i$-th loop iterates through all $v_i\in P(\{1,\dots,k\}\setminus \{\gamma(v_1),\dots, \gamma(v_{i-2})\},v_{i-1})$.
By construction, none of the considered sets can be empty and each $(v_1,\dots,v_k)$ considered is a unique $\{1,\dots,k\}$-colorful path because it neither shares a color with $v_1, \cdots, v_{i-2}$ nor with $v_{i-1}$.

Note that for a given coloring $\gamma$ and path $\pi$, it is easy to decide if the algorithm described above  will output $\pi$.
To do this, simply check whether $\pi$ is $\{1,\dots,k\}$-colorful for $\gamma$.

Next we will discuss the colorings we use in more detail.
For this, we make use of \textit{perfect hash families}.
A perfect hash family is a set of colorings $\Gamma$ (using $k$ colors) such that for any subset $\pi\subseteq V$ of size $k$ there exists a coloring $\gamma \in \Gamma$ that colors the vertices $v\in \pi$ with unique colors, i.e., if $u, v\in \pi, u\neq v$, then $\gamma(u)\neq \gamma(v)$.
According to \cite{DBLP:conf/focs/NaorSS95}, we can find such a family of colorings $\Gamma$ of size $e^k k^{O(log k)} \log n$ in time $e^k k^{O(log k)} n \log n$, i.e., in $\fpt$-time.
Let $\Pi_\gamma$ be the $\{1,\dots, k\}$-colorful paths for $\gamma \in \Gamma$.
Note that every path in $G$ of length $k$ is $\{1,\dots, k\}$-colorful for some $\gamma \in \Gamma$.

In total, we have that the set of all paths of length $k$ is $\bigcup_{\gamma\in \Gamma}\Pi_\gamma$, we can enumerate each $\Pi_\gamma$ in $\DelayFPT$, and we can check membership of $\Pi_\gamma$ in $\fpt$ (actually also polynomial time).
Now let us apply the technique developed in \cite{home-made:Strozecki10} for avoiding duplicates.
To that end, we do the following:
We take an arbitrary order of the colorings $\{\gamma_1,\dots,\gamma_{|\Gamma|}\}=\Gamma$.
Then, we start enumerating $\Pi_{\gamma_1}$ until the algorithm described above finds its first solution $\pi_1$.
At that point, we pause the enumeration of $\Pi_{\gamma_1}$ and check whether one of $\Pi_{\gamma_2}, \dots, \Pi_{\gamma_{|\Gamma|}}$ contains $\pi_1$.
If this is the case, we discard $\pi_1$ and simply start enumerating $\Pi_{\gamma_2}$ up to the first solution $\pi_2$.
At that point, we also pause the enumeration of $\Pi_{\gamma_2}$ and check whether one of $\Pi_{\gamma_3}, \dots, \Pi_{\gamma_{|\Gamma|}}$ contains $\pi_3$.
We continue this process recursively.
That is, we are searching for the first $\Pi_{\gamma_i}$ whose next solution is $\pi_i$ but which is not part of any $\Pi_{\gamma_{i+1}},\dots, \Pi_{\gamma_{|\Gamma|}}$.
Now we can output $\pi_i$ and \textit{not} output $\pi_1,\dots, \pi_{i-1}$ (they are discarded from $\Pi_{\gamma_1}, \dots, \Pi_{\gamma_{i-1}}$). 
Note that the solutions $\pi_1,\dots, \pi_{i-1}$ will be output later using a different coloring $\gamma\in \Gamma$.
We apply this recursive procedure until for every $\Pi_{\gamma_i}$ every $\pi_i\in \Pi_{\gamma_i}$ was either output or discarded.

Each recursive procedure call requires at most $|\Gamma|$-many calls to a $\DelayFPT$ enumeration procedure and $|\Gamma|$-many $\fpt$ membership calls.
Thus, in total, \lpth is in $\DelayFPT$.
\end{proof}

\noindent
{\bf Generalization and formalization.}
Intuitively, our union enumeration method seeks to first split the solution space into subsets of solutions such that these subsets cover the entire solution space but they are, in general, not disjoint. To combine the enumeration algorithms of these subsets into an enumeration of the entire solution space while avoiding duplicates, we assign an instance $(I,k)$ a finite number of identifiers  (e.g. certain colorings on a graph) $c(I,k)=\{\gamma_1, \cdots \gamma_l\}$, and we need a function $f$ that enumerates a subset of solutions $P_i\subseteq Q(I,k)$ for each identifier $\gamma_i$ without taking too long between consecutive solutions. While we require that these $P_i$ cover the entire set $Q(I,k)$, there may be pairwise overlaps. As has been discussed above, in order to guarantee 
$\fptd$, we use a technique developed in \cite{home-made:Strozecki10} for enumerating elements of a union of sets without duplicates. Besides the efficient enumerability of each set in the union (in our case $\fptd$), this technique requires an (in our case $\fpt$) function $g$ to decide whether a solution $S$ belongs to the solution set $P_i$ of a given identifier $\gamma_i$ in order to quickly check if there is any overlap between solution sets.  
We formalize this idea in the following definition.

\begin{defi}
\label{def:union}
    For an enumeration problem $Q$, we call a tuple of functions $(c,f,g)$  a \emph{union enumeration algorithm} 
    if, for some computable function $h\colon \N \to \N$ and polynomial $p$, the following conditions hold:
    \begin{itemize}
     \item $c\colon \Sigma^* \times \N \to \mathcal{P}(\Sigma^*)$ and %
     $c(I,k)$ can be computed in time
     $O(h(k)p(|(I,k)|))$, %
        \item $f\colon \Sigma^* \times \N \times \Sigma^* \to \mathcal{P}(\Sigma^*)$ is such that for all  $(I,k,\gamma)\in \Sigma^* \times \N \times \Sigma^*$:
        \begin{itemize}
        \item $\bigcup_{\gamma \in c(I,k)} f(I,k,\gamma) = Q(I,k)$,
            \item for every $\gamma \in c(I,k)$, the delay of $f(I,k,\gamma)$ is $O(h(k)p(|(I,k)|))$,
        \end{itemize}
        \item $g\colon \Sigma^* \times \N \times \Sigma^* \times \Sigma^* \to \{0,1\}$ is such that for all  $(I,k,\gamma, S) \in \Sigma^* \times \N \times \Sigma^* \times \Sigma^*$:
        \begin{itemize}
        \item $g(I,k,\gamma,S) = 1$ if and only if ${S \in f(I,k,\gamma)}$,
            \item $g(I,k,\gamma, S)$ can be  computed in time $O(h(k)p(|(I,k)|))$.
        \end{itemize}
    \end{itemize}
\end{defi}

It remains to make explicit how to efficiently enumerate a union of non-necessarily disjoint sets. 
This is done in the proof of the following proposition (for details, see 
\iflncs
\cite{DBLP:journals/corr/abs-2509-11929}).
\else
Appendix~\ref{app:UnionEnumeration}).
\fi

\begin{prop}
\label{prop:UnionEnumeration}
    Let $Q$ be a parameterized enumeration problem. If $Q$ has a union enumeration algorithm $(f,g,c)$, then $Q$ is in $\DelayFPT$.
\end{prop}

To fit our $\lpth$ algorithm from the proof of Proposition \ref{prop:lpth}
into this framework, we would first compute a perfect hash family $\Gamma$ and,
for given instance $I = G = (V,E)$, set 
$c(G,k) = \Gamma$. Moreover, we define $f(G,k, \gamma)$ as the 
set of $\{1, \dots, k\}$-colorful paths, and $g(G,k, \gamma, \pi)$ to indicate if 
$\pi$ is a  $\{1, \dots, k\}$-colorful path. As discussed above, all of these steps can be done in $\fptd$, resp. $\fptt$.

\section{Iterative Compression}
\label{sec:IterativeCompression}
Iterative compression is another classic method to develop $\fpt$ decision algorithms. %
Intuitively, the basic idea 
is to start with a solution $S_\text{sub}$ of a small sub-instance $(I_\text{sub},k)$, and to gradually increase the size of the sub-instance while maintaining a solution until we arrive back at the original instance $(I,k)$.
Maintaining the solution $S_\text{sub}$ and growing the instance $(I_\text{sub},k)$ is done in two steps:
(1) We start by growing $I_\text{sub}$ to $I_\text{next}$ while, at the same time, we increase the parameter $k$ to $k'$.
Here, the increase in the parameter and the access to a solution of the smaller instance should help us to efficiently find a solution $S_\text{next}'$ of $(I_\text{next},k')$.
(2) Then, we \textit{compress} $S_\text{next}'$ down to a solution $S_\text{next}$ of $(I_\text{next},k)$.

We thus naturally arrive at a solution of $(I,k)$ if it exists.
To transform a decision iterative compression algorithm into an enumeration one, it seems natural to simply transform the compression step into an enumeration algorithm.
That is, instead of compressing $S_\text{next}'$ down to a single solution of $(I_\text{next},k)$, we should aim at an algorithm that manages to enumerate \textit{all} solutions of $(I_\text{next},k)$ with the help of $S_\text{next}'$.
In fact, the adaptation is only necessary in the very last compression step, i.e., where $(I_\text{next},k)$ is our original instance.

We will illustrate how an iterative compression decision algorithm can be transformed into an enumeration algorithm with the well-known 
$\vc$, i.e., 
given a graph $G = (V,E)$, we are looking for vertex covers of size $\leq k$ where $k$ is the parameter.
Recall, that a set $S \subseteq V$ is a vertex cover if $G[V\setminus S]$ only contains isolated vertices.

\begin{prop}
\label{prop:VcIterativeCompression}
    $\vc$ is in $\DelayFPT$.
\end{prop}
\begin{proof}
Let $(G,k)$ be an instance of $\vc$.
For the graph $G= (V,E)$ with $V = \{v_1, \cdots, v_n\}$, let $G_i$ denote the subgraph induced by the vertex set $\{v_1, \dots, v_i\}$.
These graphs $G_i$ will be the sub-instances considered by the iterative compression algorithm.
We can do the following:

First, note that, trivially, $\emptyset$ is a vertex cover of $G_0 = (\emptyset, \emptyset)$ and, hence, $(G_0,k)$ can be our starting instance and $\emptyset$ our initial solution $S_0$.
Then, inductively, for a vertex cover $S_i$ of $G_i$ of size $\leq k$, the set $S'_{i+1} = S_i \cup \{v_{i+1}\}$ is a vertex cover of $G_{i+1}$ of size $\leq k+1$.
To compress $S'_{i+1}$ to a vertex cover of size $k$ of $G_{i+1}$, we observe that we can partition the set of 
all vertex covers of $G_{i+1}$ according to their intersection with $S'_{i+1}$.
Conversely, given a  vertex cover $S$  of $G_{i+1}$, we can split $S'_{i+1}$ into subsets 
$C, F \subseteq S'_{i+1}$ with $S'_{i+1} \cap S = C$ and $S'_{i+1} \setminus  S = F$.
Let $N_{G_{i+1}}(F)$ be the union of the neighbors of the vertices $F$ in $G_{i+1}$, excluding $F$ itself. 
Clearly, if a vertex cover $S$ of $G_{i+1}$ does not contain any of the vertices in $F$, then it must contain all of $N_{G_{i+1}}(F)$. In fact, there exists a vertex cover $S$ having $S \cap S_{i+1} = C$ if and only if $C \cup N_{G_{i+1}}(F)$ is a vertex cover.
Hence, to compress $S'_{i+1}$, it suffices to iterate over all partitionings $S'_{i+1} = C \cup F$.
Note that there are only up to $\fpt$-many partitions.
For each partition, we simply have to check whether $S_{i+1} = C \cup N_{G_{i+1}}(F)$ is a vertex cover of size at most $k$.

As long as $i+1 < n$, we can use {\em the first} vertex cover $S_{i+1}$ of size $\leq k$ of 
$G_{i+1}$ as the compressed solution. (The decision procedure would also do this at $i+1=n$.)
Thus, we proceed with $S_{i+1}$ to $G_{i+2}$. 
In contrast, at $i+1 = n$, we now have to output \textit{all} vertex covers of $G_{i+1}=G$ of size $\leq k$.
The good thing is that the compression procedure nicely allows for that.
To do so, we simply go over all $S_n=S_{i+1} = C \cup N_{G}(F)$ that are vertex covers of size $\leq k$.
Then, for each such $S_n$, we output every $\hat{S}_{n} \supseteq S_{n}$ where $ \hat{S}_{n}\setminus S_{n} \subseteq V\setminus S'_{n}$ and $|\hat{S}_{n}|\leq k$.  
That is, we output $S_{n}$ and all supersets of $S_{n}$ with 
$\leq k$ elements,
obtained by only adding vertices from $V\setminus S'_{n}$.
This guarantees that we will never produce duplicates. 
Moreover, it is easy to see that this enumeration can be done with $\fptd$.
Thus, in total, this enumeration version of iterative compression solves $\vc$ with $\fptd$.
\end{proof}

\noindent
{\bf Generalization and formalization.}
To transform decision algorithms based on iterative compression into enumeration 
algorithms, we describe the \textit{growing} and \textit{compression} steps very abstractly.
Concretely, we \textit{grow} the instances $I$ one character $\sigma\in \Sigma$ at a time.
Thus, from a solution $S_{\text{sub}}$ of $(I_{\text{sub}},k)$ we want to be able to determine a solution $S'_{\text{next}}$ of $(I_{\text{sub}}\cdot \sigma, k+1)$ in the growing step, and then compress it down to a solution $S_{\text{next}}$ of $(I_{\text{sub}}\cdot \sigma, k)$.
These two steps are performed by the functions $g$ and $c$, respectively, in the definition below.
Note that for the iterative compression method to work, we need solvability to be preserved when taking sub-instances.
Hence, we restrict ourselves here to enumeration problems $Q$ which 
are such that for any character $\sigma \in \Sigma$, if $Q(I \cdot \sigma,k) \neq \emptyset$, then $Q(I,k) \neq \emptyset$. Note that, while we formalize these steps as growing and compressing $k$ by 1 and $I$ by a single character, these steps obviously also work for other ways of growing and compressing instances and solutions.

\begin{defi}
\label{def:comp}
    Let $Q$ be an enumeration problem such that for any character $\sigma \in \Sigma$, if $Q(I \cdot \sigma,k) \neq \emptyset$, then $Q(I,k) \neq \emptyset$.
    We call a tuple of functions $(g,c)$ an {\em iterative compression algorithm} for $Q$,
    if
    the following conditions hold: 
    \begin{itemize}
        \item $g\colon \Sigma^* \times \Sigma \times \N \times \Sigma^* \to \Sigma^*$ is an $\fptt$ function such that, for $(I,k,S) \in Q$, 
        we have $g(I, \sigma, k,S) \in Q(I \cdot \sigma, k+1)$ for any $\sigma \in \Sigma$,
        \item $c\colon\Sigma^* \times \N \times \Sigma^* \to \mathcal{P}(\Sigma^*)$ is a function which, for any   $(I,k+1,S) \in Q$, enumerates all $S' \in Q(I,k)$ with $\fptd$. 
    \end{itemize}
\end{defi}

An $\fptd$ algorithm based on this definition works as follows. 
We first restrict our problem to a small part, where finding a solution is trivial (e.g., a small or empty subgraph for graph problems or a constant sized $I$ in general).
Then we apply two steps repeatedly: The first step, corresponding to the function $g$, is the growing step, in which we take our reduced instance and solution and reintroduce deleted information one step at a time, increasing both instance and parameter such that we can easily find a new solution $S$ to this new instance $(I',k+1)$.
The second step, corresponding to the function $c$,  is the compression step, where we use $S$ to find all solutions in $Q(I',k)$ (and stop if none are found). We continue rebuilding and then compressing (until the first solution is found) until we arrive at our initial instance and can now output all solutions in $Q(I,k)$.
Note that we could actually allow the parameter in each iteration of $g$ to grow by a function $b(k)$,
provided that $c$ can compress the solution down to size $k$.

\begin{prop}
\label{prop:IterativeComp}
    Let $Q$ be a parameterized problem. If $Q$ has an iterative compression algorithm $(g,c)$, then $Q$ is in $\DelayFPT$.
\end{prop}

To relate our algorithm from the proof of Proposition \ref{prop:VcIterativeCompression}
to this framework, we have used a slightly more general growing function $g$, which 
allowed us to extend a solution by yet another vertex (and not just a single character). The function $c$ takes a vertex cover $S$ of size $k+1$ as input and 
enumerates all vertex covers of size $\leq k$ of the same graph. This enumeration
is done by partitioning the resulting vertex covers according to their
intersection with $S$.

\newcommand{\dist}{\mathrm{dist}}
\newcommand{\steiner}{\mathcal{S}}

\section{Dynamic Programming}
\label{sec:DynamicProgramming}

Dynamic programming (DP) is a widely used method
for solving problems by 
splitting instances recursively into ever smaller sub-instances, 
and ultimately inductively combining the solutions 
of smaller sub-instances to solutions of bigger ones. 
DP also has many applications in 
the area of fixed-parameter algorithms, see, e.g., the \lpth problem revisited in Section~\ref{sec:UnionEnumeration}. 
Above all, 
DP algorithms use a table $T$, whose size is bounded by $O(f(k) n^{O(1)})$, to 
represent the sub-instances of interest. 
Then the minimum machinery needed are 
methods for initializing $T$ and for propagating entries of $T$ from smaller sub-instances to bigger ones, ``backlinks'' (from a bigger sub-instance to smaller ones) to keep track of the ``provenance'' of each entry in $T$, and another method for propagating partial solutions from smaller sub-instances to bigger ones.
The DP-based enumeration algorithm discussed in Section~\ref{sec:UnionEnumeration} 
was rather straightforward, as every sub-instance 
had a unique backlink.
However, in general, quite some additional care is required to avoid duplicates.

To illustrate the additional measures often needed to get a proper 
$\fpt$-enumeration algorithm via DP, we  consider the {\sc Steiner Tree} problem.
That is, we are given an undirected graph $G= (V, E)$ with an edge weight function assigning positive weights to each edge and a subset $K \subseteq V$ of ``terminal nodes''. 
A \emph{Steiner tree} for $K$ in a graph $G$ is a connected subgraph $H$ of $G$ that contains $K$. The weight of a subgraph $H$ is the sum of its edge weights. 
Then the enumeration problem consists of enumerating all Steiner trees with minimum weight. 
The decision problem (i.e., does there exist a Steiner tree of weight $\leq B$ for some bound $B$) 
is known to be $\fpt$ w.r.t.\ $k = |K|$. It has an elegant $\fpt$-algorithm (see \cite[Section 6.1.2]{DBLP:books/sp/CyganFKLMPPS15}), which works as follows:  

W.l.o.g., we may assume that $|K| > 1$ and that each terminal has degree 1 in $G$ and its neighbor is a non-terminal.
The first assumption is justified as otherwise $(K,\emptyset)$ is the only minimal weight Steiner tree and the second assumption is justified as otherwise we can simply modify $G$ to have this property.
That is, for each $k\in K$ a new neighbor $v_k\not\in V$ to $v_k$ with an arbitrarily weighted edge.
Then, $K':=\{v_k\mid k\in K\}$ are the new terminal nodes for $G'=(V\cup K', E \cup \{kv_k\mid k\in K\})$.
Note that $(G,K)$ and $(G',K')$ admit the same minimal weight Steiner tree (modulo adding/removing $K'$).

The data structure to be propagated by the DP algorithm consists 
of entries $T[D,v]$ with $D \subseteq K$ and $v \in V \setminus K$. Then $T[D,v] = w$ means that the minimum weight of a Steiner tree of 
$D \cup \{v\}$ in $G$ is $w$. The idea of the vertex $v$ (which of course does not necessarily have degree 1 in $G$) is 
that we combine Steiner trees of ``sub-instances'' (i.e., two disjoint subsets $D_1,D_2 \subseteq K$) by merging them at vertex $v$ to get a Steiner tree of $D_1 \cup D_2$. In the sequel,
we will refer to Steiner trees with minimal weight simply as ``Steiner trees'' (omitting the minimality).
Moreover, we will denote the set of Steiner trees of $D \cup \{v\}$ in $G$ as $\steiner (D,v)$ and call them 
the ``solutions'' of $T[D,v]$. The set of Steiner trees of $G$ is obtained as the union of those $\steiner (K,v)$, such that 
$T[K,v]$ is minimum over all $v$.

The data structure $T$ is filled in as follows: as initialization step of $T$, 
we compute all entries of the form $T[\{t\},v] = \dist(t,v)$, where $t \in K$, $v \in V\setminus K$,
and $\dist(t,v)$ denotes the distance 
(i.e., the minimal weight of a connecting path) 
from $t$ to $v$. 
Then, existing entries of $T$ are combined recursively via the formula
\begin{equation}
\label{eq:steinerRecursion}
T[D,v] = \min_{\parbox[t]{42pt}{\scriptsize  \mbox{$u \in V \setminus K$} \\ 
   \mbox{$\emptyset \neq D' \subsetneq D$}} }
   \{T[D',u] + T[D\setminus D',u] + \dist(u,v) \}  
\end{equation}
\noindent
Intuitively, a Steiner tree for $D \cup \{v\}$ is obtained by merging Steiner trees
for $D' \cup \{u\}$ and $(D \setminus D') \cup \{u\}$ 
at vertex $u$ and then adding a minimal weight connecting path from 
$u$ to $v$. This includes the special case that $u = v$. 
It is shown in \cite[Lemma 6.2]{DBLP:books/sp/CyganFKLMPPS15} that
one thus gets all entries $T[D,v]$ with $\emptyset \neq D \subseteq K$
and $v \in V \setminus K$. Hence, deciding if a 
Steiner tree of weight $\leq B$ exists comes down to checking 
if an entry $T[K,v] \leq B$ exists.
To also output a solution in case of a yes-instance, one has to maintain backlinks 
for each application of the above recursion and assemble a Steiner tree of $D \cup \{v\}$ by merging 
two subtrees and adding a shortest path as described above. 
Maintaining such backlinks and 
combining solutions from smaller sub-instances to solutions of bigger ones 
is routine in the area of DP algorithms. 

However, if we want to enumerate {\em all} Steiner trees, 
we are immediately faced with the challenge of avoiding {\em duplicates}.
Indeed, an entry $T[D,v]$ can possibly be 
obtained by combining entries $T[D',u_1]$ and $T[D \setminus D',u_1]$ 
but also by combining entries $T[D'',u_2]$ and $T[D \setminus D'',u_2]$ 
with $D' \neq D''$
and/or $u_1 \neq u_2$.
Likewise, the solutions of $\steiner[D,v]$
(i.e., Steiner trees of $D \cup \{v \} ]$) 
obtained by combining solutions of $\steiner[D',u_1]$ and $\steiner[D \setminus D',u_1]$  
and the solutions of  $\steiner [D,v]$
obtained by combining solutions of $\steiner[D'',u_2]$ and $\steiner[D \setminus D'',u_2]$  
are not necessarily disjoint.
Hence, recursively iterating through all possible 
backlinks, starting from the
entries $\steiner[K,v]$ with minimal weight, 
would inevitably produce duplicates.

To achieve unique provenance, intuitively, we have to make sure that for each Steiner tree $S\in \steiner [D,v ]$ there is a unique partitioning $D = D' \cup D\setminus D'$ and a unique merge node $u\in V\setminus K$.
This requires several modifications 
compared with the decision algorithm, namely: (1) we 
(arbitrarily) fix a strict order on the terminals in $K$,
(2) we extend entries in table $T$ to $T[D,v,b]$, where 
$b$ is a bit indicating if, in the Steiner trees of $D \cup \{v\}$ represented by $T[D,v,b]$, 
$v$ has degree 1 (indicated by  $b=1$) or degree $> 1$ (indicated by $b = 0$), 
and (3) we impose the restriction that the recursion in Equation (\ref{eq:steinerRecursion}) may only be applied if $D'$ contains the smallest 
terminal in $D$ (according to the chosen order) and the left operand 
is of the form $T[D',u,1]$, i.e., $u$ has degree 1 in all Steiner trees thus represented. 
Analogously, we adapt the definition of $\steiner[D,v,b]$.

We claim that, with these modifications, we indeed achieve a unique provenance of each entry $\steiner[D,v,b]$
and that we can recurse through the backlinks of these entries so as to output all Steiner trees
of a given problem instance with $\fptd$. The formal proof of this claim is given in 
\iflncs
\cite{DBLP:journals/corr/abs-2509-11929}.
\else
Appendix~\ref{app:DynamicProgramming}.
\fi
We thus get:

\begin{proposition}
{\sc Steiner Tree} is in $\DelayFPT$.
\end{proposition}

\section{Conclusion}
\label{sec:Conclusion}

We have revisited several well-established methods for designing $\fpt$ decision algorithms 
in order to transform them into $\fptd$ enumeration algorithms.
We have illustrated by various examples the subtle differences and the special care required for enumeration algorithms -- 
in particular, the need to obtain all solutions and to avoid duplicates.

While we have explored the design of $\fptd$ enumeration algorithms starting from $\fptt$ techniques, one could also proceed in the
other direction and  adapt enumeration techniques to the parameterized setting. For instance, the well-known reverse-search method \cite{DBLP:journals/dam/AvisF96},
as well as the new  and very promising proximity search method \cite{conte2022proximity} would probably be worth a systematic investigation from the $\FPT$ point of view.

\section*{Acknowledgements}

The work of N.~Creignou has 
been supported by the Agence Nationale de la Recherche (ANR) project PARADUAL [ANR-24-CE48-0610-01].
The work of T.~Merkl, R.~Pichler, and D.~Unterberger has been supported by 
the Vienna Science and Technology Fund (WWTF) [10.47379/ICT2201].

 \bibliographystyle{splncs04}
 \bibliography{refs}

\iflncs
\else

\clearpage

\appendix

\section{More Details on Section \ref{subsec:boundedSearchTree}}
\label{app:boundedSearchTree}

We first provide a proof sketch of Proposition~\ref{prop:bounded_search_tree}.

\begin{proof}[Proof sketch of Proposition~\ref{prop:bounded_search_tree}] 
Given an instance $(I,k)\in  \Sigma^* \times \N$ the following algorithm enumerates all solutions in 
 $Q(I,k)$ with $\fptd$:
 \begin{itemize}
        \item Calculate $\gamma(I,k)$. 
        \item If $\gamma(I,k) =  0$, apply $f_0$ to enumerate $Q(I,k)$.
        \item Otherwise, recursively apply this algorithm to each $(I',k') \in f(I,k)$  depth first.
    \end{itemize}

By construction, this method correctly enumerates all solutions in $Q(I,k)$. 
The size of the search tree is $O(b(k)^{h(k)})$. At each node $(I',k')$ we can either process in $\fptt$ or enumerate its solution set with $\fptd$. Since by assumption all nodes satisfy $|(I',k')|\le|(I,k)|$,  so the delay is $O(b(k)^{h(k)}l(k)|(I,k)|^c)$ for some constant $c$ and some computable function $l$.
\end{proof}

We now discuss in detail, how the 
{\sc Feedback Vertex Set in Tournaments} (FVST, for short) enumeration problem
fits 
Definition~\ref{def:bounded_tree_partition}:
Recall the additional data structures $C,F$ that we introduced in the proof of 
Proposition~\ref{prop:fvst}. We now integrate these two  vertex sets into the
problem definition and define the {\sc Generalized Feedback Vertex Set in Tournaments} (GFVST, for short)
problem. 

In  GFVST, we are given a tournament $T=(V,A)$, a parameter $k \in \N$, and two sets $C,F \subseteq V$ having $C \cap F = \emptyset$. For such an instance $((T,C,F),k)$, we want to enumerate all feedback vertex sets $S$ of $T$ with $C \subseteq S$, $S \cap F = \emptyset$, and $\# (S\setminus C) \leq k$. Clearly, FVST is the special case of GFVST with instances $((T,C,F),k)$ satisfying 
$C = F = \emptyset$.

In case of the $\gfvst$, 
this means that $f$ is the method by which we split (and adjust $k$, $C$ and $F$ accordingly):
To ensure that the sizes of the generated instances do not increase, i.e., that for all $(I',k') \in f(I,k)$, $|(I',k')| \leq |(I,k)|$, we can encode $C$ and $F$ as binary characteristic words of length $|V|$.
Such ``small'' increases in size can always be avoided by appropriate representations and will be implicitly assumed in the sequel
and  $b(k) =7$. Now if (1) $T-C$ contains a cycle and (2) the subgraph of $T$ induced by $F$ does not contain a cycle,
we set $\gamma((T,C,F),k) = k$. Otherwise, 
$\gamma(I,k) = 0$, i.e., we efficiently output the solution set, which is empty if (2) is violated. \\

\section{More Details on Section \ref{subsec:Flashlight}}
\label{app:Flashlight}

We first provide a proof of Proposition~\ref{prop:flashlight_partition}.
   
\begin{proof}[Proof of  Proposition~\ref{prop:flashlight_partition}] 
  Given an instance $(I,k)\in  \Sigma^* \times \N$  we apply the following algorithm:
   \begin{itemize}
        \item Compute $f(I,k)$. If $f(I,k)=\emptyset$ then stop.
        \item Calculate $\gamma(I,k)$. 
        \item If $\gamma(I,k) =  0$, apply $f_0$ to enumerate $Q(I,k)$.
        \item Otherwise, recursively apply this algorithm to each $(I',k') \in f(I,k)$  depth first.
    \end{itemize}
 This algorithm  partitions the solution space in each iteration and correctly enumerates all solutions.
 When traversing the search tree depth first we reach an instance having a solution and such that $\gamma(I,k) =0$   after at most $h(k)p(|(I,k)|)$ iterations of $f$ and $\gamma$.  Hence the number of steps in between the output of two successive solutions  is bounded by $2h(k)p(|(I,k)|)$. 
    Since the functions $f$ and $\gamma$ are  $\fptt$ and since $f_0$ has $\fptd$, we finally get $\fptd$.
\end{proof}

We now discuss in detail, how the 
\clst enumeration problem fits Definition \ref{def:flash_partition}.
To this end, we consider the 
\clstp problem
as a generalized version of the original  problem, 
where, we are additionally given a 
string $\omega \in \Sigma^*$  
and we want to find all center strings that contain $\omega$ as prefix.
Clearly, $\clst$ is the special case of \clst \textsc{with Prefix}
where the prefix $\omega$ is the empty string.

For \clstp, 
$f$ is the method by which we generate all $((X,\omega \cdot \sigma),k)$ and then check if a solution exists, only keeping those instances with solutions. As $\gamma$ bounds the branching depth, we need $\gamma((X,\omega),k) = L -w$ where $L$ is the length of the strings and $w$ is the length of $\omega$, as each step increases the length of the prefix by $1$. Note that $k$ never changes, thus $k' \leq k$ for all child instances, and we can ensure $|(I',k')| \leq |(I,k)|$ by filling the prefix with a dummy character for each missing letter. \\

\section{More Details on Section \ref{subsec:SolutionSearch}}
\label{app:SolutionSearch}

We provide a proof sketch of Proposition~\ref{prop:solution_partition}.

\begin{proof}[Proof sketch of Proposition~\ref{prop:solution_partition}]
The solution search partition algorithm $(f,g)$  defines a tree that implicitly contains all solutions exactly once. 
The root of the tree is the original instance $(I,k)$ and it gives rise to a first solution using the function $g$  in $\fptt$. The children of each instance in this tree are computed by the function $f$, which is also $\fpt$. 
Moreover, the size of any instance explored through our algorithm is bounded by the size of the original instance.  
Thus, in order to enumerate all solutions, it is sufficient to traverse the tree by a depth-first search. Since both $f$ and $g$ are  $\fpt$,  the cost of each iteration is $\fpt$. 
\end{proof}

\section{More Details on Section \ref{sec:UnionEnumeration}}
\label{app:UnionEnumeration}

We first provide a proof sketch of Proposition~\ref{prop:UnionEnumeration}

\begin{proof}[Proof of Proposition~\ref{prop:UnionEnumeration}]
    Given an instance $(I,k)$ of $Q$, we enumerate the solutions $Q(I,k)$ with $\fptd$ in the following way. 
    First, we compute in $\fptt$ the whole set $\{\gamma_1,\dots, \gamma_l\} = c(I,k)$ of size $O(h(k)p(|(I,k)|))$.
    Then, we run $f(I,k,\gamma_1)$ until it finds its next solution $S_1$.
    At that point, we pause $f(I,k,\gamma_1)$ and check whether one of $f(I,k,\gamma_2), \dots, f(I,k,\gamma_l)$ will at some point output $S_1$ by running $g(I,k,\gamma_2,S_1),\dots, g(I,k,\gamma_l,S_1)$.
    If this is the case, we simply run $f(I,k,\gamma_2)$ until it outputs its next solution $S_2$, which is possibly different from   $S_1$.
    (If $f(I,k,\gamma_2)$ has already been processed fully, we simply proceed to $f(I,k,\gamma_3)$.)
    At that point, we also pause $f(I,k,\gamma_2)$ and check whether one of $f(I,k,\gamma_3),\dots,f(I,k,\gamma_l)$ will at some point output $S_2$ by running $g(I,k,\gamma_3,S_2), \allowbreak \dots, g(I,k,\gamma_l,S_2)$ and so on.
    That is, we proceed through $f(I,k,\gamma_1), \dots, f(I,k,\gamma_l)$ in this manner until we find a $f(I,k,\gamma_i)$ whose next solution $S_i$ will not be output by any of $f(I,k,\gamma_{i+1}), \dots, f(I,k,\gamma_l)$.
    Now we can output $S_i$ and \textit{not} output $S_1,\dots, S_{i-1}$. 
    Note that the solutions $S_1,\dots, S_{i-1}$ will be output later using different identifiers. 
    (Indeed, the very reason why we have iterated through solutions up to $S_i$ is 
    that each $S_r$ with $r < i$  is also contained 
    in $f(I,k,\gamma_s)$ for some $s$ with $r < s \leq i$.) 
    After this, we continue running $f(I,k,\gamma_1)$ again in the same manner.
    
    We repeat this procedure until $f(I,k,\gamma_1)$ has no next solution.
    At that point, we proceed with $f(I,k,\gamma_2)$ in the same manner.

    First note that this procedure outputs each solution exactly once, as for each $f(I,k,\gamma_i)$, we output the solutions $f(I,k,\gamma_i)\setminus \bigcup_{j>i} f(I,k,\gamma_j)$ and $Q(I,k) = \bigcup_i(f(I,k,\gamma_i)\setminus \bigcup_{j>i} f(I,k,\gamma_j))$.

    As far as the delay is concerned, note that, in the worst case, we have to wait for each $f(I,k,\gamma_i)$ to produce its next solution and evaluate each $g(I,k,\gamma_i,S)$ for 
    at most $l-1$ different solutions $S$.
    Thus, the delay is in $O((h(k)p(|(I,k)|))^3)$.
\end{proof}

We now discuss in detail, how the 
$\lpth$ enumeration problem fits 
Definition~\ref{def:union}.
In the $\lpth$ problem, we are given a graph $G = (V,E)$ with $V=\{v_1,\dots,v_n\}$ and an integer $k$.
First recall that it is possible to compute a perfect hash family $\Gamma$ using $k$ colors of size $e^k k^{O(log k)} \log n$ in time $e^k k^{O(log k)} n \log n$, i.e., in $\fpt$-time \cite{DBLP:conf/focs/NaorSS95}.
Thus, an algorithm that does so can be used as the function $c$, i.e., $c(G,k)=\Gamma$.

Then, next, we need to assign subsets of the solutions to the different identifiers $\Gamma$ via the function $f$.
Naturally, for $\gamma\in \Gamma$ let $f(G,k,\gamma)$ be equal to the $\{1,\dots,k\}$-colorful paths $\Pi_\gamma$ for $\gamma$.
Crucially, since $\Gamma$ is a perfect hash family, we have $\bigcup_{\gamma \in \Gamma}\Pi_\gamma = \Pi,$ where $\Pi$ are all the paths of length $k$ in $G$.
Furthermore, we have seen that it is possible to enumerate all sets $\Pi_\gamma$ with $\fptd$.

Lastly, the function $g$ simply is defined as
$$g(G,k,\gamma,\pi) = \begin{cases}
    1 & \pi \text{ is } \{1,\dots,k\}\text{-colorful for } \gamma,\\
    0 & \text{otherwise},
\end{cases}$$
where $\pi$ is a candidate path.

\section{More Details on Section \ref{sec:IterativeCompression}}

We first provide a proof sketch of Proposition~\ref{prop:IterativeComp}.

\begin{proof}[Proof sketch of Proposition~\ref{prop:IterativeComp}]
    For $I = \sigma_1 \cdots \sigma_{|I|}$, let $I_i = \sigma_1 \cdots \sigma_i$ (we define $I_0 : = \varepsilon$).
    Following the algorithm already sketched above, we do the following:
    We start with a trivial solution $S_0$ of the empty subproblem $(I_0,k)$.
    Then we repeatedly apply $g$ to $(I_i,k)$ and a solution $S_i$ to find a solution $S'_{i+1}$ of $(I_{i+1},k+1)$, and then run $c$ until it gives its first solution $S_{i+1}$ for $(I_{i+1},k)$.
    If any $(I_{i+1},k)$ has no solution, we can assert that $(I,k)$ has no solution either. 
    Otherwise, once $I_{i+1} = I$, run $c$ until it has output all solutions.

    By our assumption, $c$ has $\fptd$, so we only need to bound the preprocessing time. 
    For that, we execute $g$ and $c$ (until it gives its first solution) at most $|I|$ times. 
    Since $g$ is an $\fptt$ function and $c$ finds the first solution in $\fpt$ time, 
    we thus achieve $\fptd$.
\end{proof}

We now discuss in detail, how the 
iterative compression algorithm in the proof of 
Proposition \ref{prop:VcIterativeCompression} for the  $\vc$ problem
fits Definition~\ref{def:comp}.
For a graph $G=(\{v_1,\dots,v_n\},E)$, let $G_i$ denote the graph induced by the vertex set $\{v_1, \dots, v_i\}$ and for $U \subset V$, 
let $N_{G_i}(U)$ be the union of the neighbors 
of vertices $U$ in $G_i$, excluding $U$ itself. 
Trivially, $\emptyset$ is a vertex cover of $G_0 = (\emptyset, \emptyset)$. 
Then, for a vertex cover $S_i$ of $G_i$ of size $\leq k$, $g$ constructs the vertex cover $S'_{i+1} = S_i \cup \{v_{i+1}\}$ of $G_{i+1}$ of size $\leq k+1$.\footnote{Formally, Definition \ref{def:comp} requires us to grow the instance one character $\sigma\in \Sigma$ at a time, whereas we add one vertex at a time. Of course, by appropriate interpreting the encodings of the instances of $\vc$ one can easily define $g$ to handle one character at a time. However, this would unnecessarily obfuscate the crucial parts and we therefore decided against presenting it in this manner.}

For the compression step, we observe that we can partition the set of 
all vertex covers of $G_{i+1}$ according to their intersection with $S'_{i+1}$.
Conversely, given a  vertex cover $S$  of $G_{i+1}$, we can split $S'_{i+1}$ into subsets 
$C, F \subseteq S'_{i+1}$ with $S'_{i+1} \cap S = C$ and $S'_{i+1} \setminus  S = F$.
Clearly, if a vertex cover $S$ of $G_{i+1}$ does not contain any of the vertices in $F$, then it must contain all of $N_{G_{i+1}}(F)$.
Hence, to compress $S'_{i+1}$, $c$ iterates over all partitionings $S'_{i+1} = C \cup F$.
For each partition, $c$ checks whether $S_{i+1} = C \cup N_{G_{i+1}}(F)$ has size at most $k$ and is a vertex cover (i.e., if $F$ is an independent set in $G_{i+1}$).

As long as $i+1 < n$, we use $c$ only to find {\em the first} vertex cover $S_{i+1}$ of size $\leq k$ of 
$G_{i+1}$, and then use  $S_{i+1}$ to proceed with $G_{i+2}$.
In contrast, for $i+1 = n$, we use $c$ to 
output {\em all} vertex covers $S_{i+1}$ of size $\leq k$ of $G_{i+1} = G_n$.
In addition, for all such $S_{i+1}$ whose size is strictly smaller than $k$, 
$c$ also has to output every $\hat{S}_{i+1} \supseteq S_{i+1}$ where $ \hat{S}_{i+1}\setminus S_{i+1} \subseteq V\setminus F$ and $|\hat{S}_{i+1}|\leq k$. We note that, for every such 
$S_{i+1}$, $c$ can construct (and output) all vertex covers $\hat{S}_{i+1}$ 
by adding any subset of vertices from $V\setminus F$ as long as $|\hat{S}_{i+1}| \leq k$  holds.
By only adding vertices from $V\setminus F$, it is guaranteed that 
$c$ will never produce duplicates. Moreover, it is easy to see that this enumeration 
can be done with $\fptd$.

\section{More Details on Section \ref{sec:DynamicProgramming}}
\label{app:DynamicProgramming}

Recall from Section~\ref{sec:DynamicProgramming} that 
we have defined several modifications that are required in 
an enumeration algorithm for the {\sc Steiner Tree} problem
compared with the decision algorithm, namely: 
\begin{enumerate}
\item[(1)] we 
(arbitrarily) fix a strict order on the terminals in $K$,
\item[(2)] we extend entries in table $T$ to $T[D,v,b]$, where 
$b$ is a bit indicating if, in the Steiner trees of $D \cup \{v\}$ represented by $T[D,v,b]$, 
$v$ has degree 1 (indicated by  $b=1$) or degree $> 1$ (indicated by $b = 0$), 
and 
\item[(3)] we impose the restriction that the recursion in Equation (\ref{eq:steinerRecursion}) may only be applied if $D'$ contains the smallest 
terminal in $D$ (according to the chosen order) and the left operand 
is of the form $T[D',u,1]$, i.e., $u$ has degree 1 in all Steiner trees thus represented.
\end{enumerate}
Analogously, we adapt the definition of $\steiner[D,v,b]$.

It remains to show that these measures indeed make the provenance unique.
Let $S$ be an arbitrary Steiner tree in $\steiner[D,v,1]$ seen as rooted in $v$.
Thus, $v$ has degree $1$ and we can follow the descendants of $v$ until we find the first node $u\in V$ of degree $ \geq 2$ in $S$ (if there is no such node, then  $|D|=1$).
Then, there is one child $c$ of $u$ whose subtree contains the smallest terminal $d_{\min}$ in $D$.
Let $S_c, S_u$ be subtrees rooted in $u$ where $S_c$ contains $c$ and its descendants, and $S_u$ the remainder.
We denote the terminals (leaves) of $S_c$ by $D'$ and the remaining terminals by $D\setminus D'$.
Then $S_c \in \steiner [D',u,1]$ while $S_u\in \steiner[D\setminus D',u,0] \cup \steiner[D\setminus D',u,1]$.
Furthermore, $S$ is composed of $S_c,S_u$, and a minimal path from $u$ to $v$.
For $S\in \steiner[D,v,0]$, the only difference is that $u=v$.

We note that all of these choices are unique, and thus, the provenance is unique.
We could not have chosen a $u$ further down as then some edges would be counted twice in the recursion, i.e., once in a $T[\cdot, u ]$ part and once in $\dist(u,v)$.
Furthermore, as $D'$ must contain the smallest terminal and $u$ has to have degree 1 in the left sub-Steiner tree, also $D'$ cannot be chosen differently.

Towards an $\fptd$-enumeration algorithm for the {\sc Steiner Tree} problem, the only obstacle remaining 
is to discuss how to recurse through the backlinks.
To that end, we start by picking an arbitrary terminal $k\in K$ and its unique neighbor $v_k\in V\setminus K$.
All Steiner trees have to go through $v_k$ and are, thus, exactly $\steiner[K,v_k,0]$.
Then, for arbitrary terminals $D\subseteq K$, vertices $v\in V\setminus K$, and $b\in \{0,1\}$, to enumerate the Steiner trees $\steiner[D,v,b]$ we recursively iterate through all $u$ that were used in the recursion for $T[D,v,b]$.
This gives us a split $D = D' \cup D\setminus D'$ and in 3 nested loops we simply iterate through $S'\in \steiner[D',u,1]$, $S''\in \steiner[D\setminus D',u,0] \cup \steiner[K\setminus D',u,1]$, and shortest paths $P$ between $u$ and $v$.
Each combination of $S',S''$, and $P$ constitutes a unique Steiner tree. 

\fi

\end{document}